\documentclass[11pt]{article}

\usepackage{amsmath,amssymb,graphicx,cite,color}

\numberwithin{equation}{section}

\graphicspath{{./Figs/}}

\usepackage{bm}

\newtheorem{theorem}{Theorem}[section]
\newtheorem{lemma}[theorem]{Lemma}
\newtheorem{proposition}[theorem]{Proposition}

\newenvironment{proof}{\noindent {\bf Proof} }{\endprf\par}
\def \endprf{\hfill {\vrule height6pt width6pt depth0pt}\medskip}

\definecolor{editcolor}{rgb}{1,0.87,0.87}
\makeatletter\newenvironment{editboxx}[1][.95]
{\begin{lrbox}{\@tempboxa}\begin{minipage}{#1\linewidth}\vspace{1ex}}
{\vspace{0ex}\end{minipage}\end{lrbox}\vspace{1ex}\fcolorbox{black}{editcolor}{\usebox{\@tempboxa}}\vspace{1ex}}
\makeatother

\newcommand{\vct}[1]{\bm{#1}}
\newcommand{\mtx}[1]{\bm{#1}}

\newcommand{\R}{\mathbb{R}}
\newcommand{\C}{\mathbb{C}}
\newcommand{\<}{\langle}
\renewcommand{\>}{\rangle}
\newcommand{\E}{\operatorname{\mathbb{E}}}
\renewcommand{\P}{\operatorname{\mathbb{P}}}

\newcommand{\n}{n}				
\newcommand{\m}{m}				
\newcommand{\s}{s}				
\newcommand{\OB}{\Omega}		
\newcommand{\ROB}{\mtx{R}_\OB}		

\newcommand{\Fee}{\bm{\Phi}}	

\newcommand{\FT}{\mtx{F}}		
\newcommand{\Sh}{\mtx{S}}		
\newcommand{\Mod}{\mtx{M}}		

\newcommand{\oa}{{\omega}}
\newcommand{\ob}{{\xi}}

\newcommand{\ek}{\varepsilon_k}
\newcommand{\el}{\varepsilon_\ell}

\newcommand{\eps}{\vct{\varepsilon}}

\renewcommand{\widehat}{\hat}

\newcommand{\PO}{\widehat{\mtx{P}}_{\Omega}}
\newcommand{\POab}{\widehat{P}_{\Omega}(\oa,\ob)}

\newcommand{\fxa}{\widehat{x}(\oa)}
\newcommand{\fya}{\widehat{y}(\oa)}
\newcommand{\fxb}{\widehat{x}(\ob)}
\newcommand{\fyb}{\widehat{y}(\ob)}

\newcommand{\hi}{{\widehat\infty}}

\newcommand{\fx}{\widehat{\vct{x}}}

\newcommand{\fX}{\widehat{\mtx{X}}}
\newcommand{\fY}{\widehat{\mtx{Y}}}

\newcommand{\iunit}{\mathrm{i}}

\newcommand{\spnorm}[1]{\vert\!\vert\!\vert {#1} \vert\!\vert\!\vert}

\begin{document}


\title{Restricted Isometries for Partial Random Circulant Matrices}
\author{Holger Rauhut, Justin Romberg, and Joel A.\ Tropp\thanks{H.\ R.\ is with the Hausdorff Center for Mathematics and the Institute for Numerical Simulation at the University of Bonn, Germany; email: rauhut@hcm.uni-bonn.de.  J.\ R.\ is with the School of Electrical and Computer Engineering at Georgia Tech in Atlanta, Georgia; email: jrom@ece.gatech.edu.  J.\ T.\ is with Applied and Computational Mathematics at Caltech in Pasadena, California; email: jtropp@acm.caltech.edu.
H.\ R.\ acknowledges generous support by the Hausdorff Center for Mathematics, and through the WWTF project SPORTS (MA07-004). J.\ R.\ was supported by ONR Young Investigator Award N00014-08-1-0884 and a Packard Fellowship. J.\ T.\ has been supported by ONR award N00014-08-1-0883, DARPA award N66001-08-1-2065, and AFOSR award FA9550-09-1-0643.}}

\maketitle 

\begin{abstract}
	In the theory of compressed sensing, \emph{restricted isometry} analysis has become a standard tool for studying how efficiently a measurement matrix acquires information about sparse and compressible signals.  Many recovery algorithms are known to succeed when the restricted isometry constants of the sampling matrix are small.
Many potential applications of compressed sensing involve a data-acquisition process that proceeds by convolution with a random pulse followed by (nonrandom) subsampling.  At present, the theoretical analysis of this measurement technique is lacking.  This paper demonstrates that the $\s$th order restricted isometry constant is small when the number $\m$ of samples satisfies $\m \gtrsim (\s \log \n)^{3/2}$, where $\n$ is the length of the pulse.  This bound improves on previous estimates, which exhibit quadratic scaling.
\end{abstract}

\section{Introduction}

The theory of compressed sensing~\cite{carota06,cata06,ca06-1,do06-2,foraXX,ra10}
predicts that a small number of linear samples suffice to capture all the information in a sparse vector and that, furthermore, we can recover the sparse vector from these samples using efficient algorithms.  This discovery has a number of potential applications in signal processing, as well as other areas of science and technology.

The linear data acquisition process is described by a measurement matrix.  The \emph{restricted isometry property} (RIP)~\cite{cata06,carota06-1,foraXX,ra10} is a standard tool for studying how efficiently this matrix captures information about sparse signals.  The RIP also streamlines the analysis of signal reconstruction algorithms.
It is unknown whether any deterministic measurement matrix satisfies the RIP with the optimal scaling behavior.  See, e.g., the discussion in \cite[Sec.~2.5]{ra10} or \cite[Sec.~5.1]{foraXX}.  In contrast, a variety of random measurement matrices exhibit the RIP with optimal scaling, including Gaussian matrices and Rademacher matrices~\cite{cata06,dota09,rascva08,badadewa08}.

Although Gaussian random matrices are optimal for sparse recovery, they have limited use in practice because many measurement technologies impose structure on the matrix.  Furthermore, recovery algorithms tend to be more efficient when the matrix admits a fast matrix--vector multiply.  For example, random sets of rows from a Fourier transform matrix model the measurement process in MRI imaging, and these partial Fourier matrices lead to fast recovery algorithms because they can be applied using the FFT.  It is known that a partial Fourier matrix satisfies a near-optimal RIP~\cite{cata06,rudelson08sp,ra08,ra10}; the paper \cite{ra10} contains some generalizations, and we refer to \cite{rawa10} for a variation related to recovery of sparse Legendre expansions.

Many potential applications of compressed sensing involve sampling processes that can be modeled by convolution with a random pulse.  This measurement process can be modeled using a random circulant matrix.  When we retain only a limited number of samples from the output of the convolution, the measurement process is described by a partial random circulant matrix.  This situation has been studied in several works from the compressed sensing literature, including~\cite{tropp06fi,bajwa07to,haupt10to,romberg09co,rauhut09ci}.
So far, the best available analysis of a partial random circulant matrix suggests that its restricted isometry constants do not exhibit optimal scaling.  This work describes a new analysis that dramatically improves the previous estimates.  Nevertheless, our results still fall short of the optimal scaling that one might hope for.

\subsection{Compressed Sensing}

The compressed sensing problem considers how to recover a vector $\vct{x} = (x_1, \dots, x_\n)^T \in \R^\n$ from the linear image
\[
\vct{y} ~=~ \mtx{A} \vct{x},
\]
where the matrix $\mtx{A} \in \R^{\m \times \n}$ and $\m \ll \n$. Clearly, it is impossible to reconstruct the vector $\vct{x}$ without additional prior information.  Compressed sensing introduces the extra assumption that $\vct{x}$ is $\s$-sparse, i.e., $\|\vct{x}\|_0 := \#\{\ell : x_\ell \neq 0\} \leq \s$ for some $\s \ll \n$. 
More generally, we assume that $\vct{x}$ is well-approximated by a sparse vector.

The na{\"i}ve approach of reconstructing $\vct{x}$ by solving the $\ell_0$-minimization problem,
\[
\min_{\vct{z}} \|\vct{z}\|_0 \quad \mbox{ subject to }\quad \vct{y} ~=~ \mtx{A} \vct{x},
\]
is {\sf NP}-hard \cite{na95}. Therefore, several tractable heuristics have been proposed in the literature as alternatives to $\ell_0$-minimization, most notably greedy algorithms \cite{blumensath09it,neve09-1,needell09co,tr04,fo10-2}
and $\ell_1$-minimization \cite{chdosa99,do06-2,carota06}. The latter approach
consists in solving the convex program
\begin{equation}\label{l1:min}
\min_{\vct{z}} \|\vct{z}\|_1 \quad \mbox{ subject to } \vct{y} ~=~ \mtx{A} \vct{x},
\end{equation}
where $\|\cdot \|_p$ denotes the usual $\ell_p$ vector norm.

The restricted isometry property (RIP) offers a very elegant way to analyze $\ell_1$-minimization and greedy algorithms.
Define the restricted isometry constant $\delta_{\s}$ of an $\m \times \n$ matrix $\mtx{A}$ to be the smallest positive number that satisfies
\begin{equation}
	\label{eq:RIPs}
	(1-\delta_\s)\|\vct{x}\|^2_2 ~\leq~ \|\mtx{A} \vct{x}\|^2_2 ~\leq~ (1+\delta_\s)\|\vct{x} \|^2_2
	\quad
	\text{ for all } \vct{x} \mbox{ with } \|\vct{x}\|_0\leq\s.
\end{equation}
In words, the statement~\eqref{eq:RIPs} requires that all column submatrices of $\mtx{A}$ with at most $\s$ columns are well-conditioned.
Informally, $\mtx{A}$ is said to satisfy the RIP (with order $\s$) when $\delta_{\s}$ is small (for $\s$ close to $\m$).

A number of recovery algorithms are provably effective for sparse recovery if the matrix
$\mtx{A}$ satisfies the RIP. More precisely, suppose that the matrix $\mtx{A}$ obeys \eqref{eq:RIPs} with 
\begin{equation}\label{delta:kappa}
\delta_{\kappa \s} ~<~ \delta^*
\end{equation}
for suitable  constants $\kappa \geq 1$ and $\delta^* < 1$. 
Then these algorithms precisely recover all $\s$-sparse vectors $\vct{x}$ from the measurements $\vct{y} = \mtx{A} \vct{x}$.  More generally, when the vector $\vct{x}$ is arbitrary and we acquire noisy observations
\[
\vct{y} ~=~ \mtx{A} \vct{x} + \vct{e}
\quad\text{where}\quad
\|\vct{e}\|_2 \leq \tau,
\]
these algorithms return a reconstruction $\widetilde{\vct{x}}$ that satisfies an error bound of the form
\begin{equation}\label{rec:stability}
\| \vct{x} - \widetilde{\vct{x}} \|_2~\leq~ C_1 \frac{\sigma_\s(\vct{x})_1}{\sqrt{\s}} + C_2 \tau,
\end{equation}
where $\sigma_\s(\vct{x})_1 = \inf_{\|\vct{z}\|_0 \leq \s} \|\vct{x} - \vct{z}\|_1$ denotes the error of best $\s$-term approximation in $\ell_1$ and $C_1, C_2$ are positive constants.  Table \ref{tableRIP} lists the best values available for the constants $\kappa$ and $\delta^*$
for several algorithms along with appropriate references.

\begin{table}[htdp]
\begin{center}
\begin{tabular}{|l|l|c|c|}
\hline
Algorithm & References & $\kappa$ & $\delta^*$ \\
\hline
$\ell_1$-minimization \eqref{l1:min} & \cite{carota06-1,ca08,fo10-3} & $2$ & $\frac{3}{4+\sqrt{6}} \approx 0.4652$ \\
CoSaMP & \cite{needell09co,fo10} & $4$ & $\sqrt{\frac{2}{5 + \sqrt{73}}} \approx 0.3843$\\
Iterative Hard Thresholding & \cite{blumensath09it,fo10-3} & $3$ & $1/2$\\
Hard Thresholding Pursuit & \cite{fo10-2} & $3$ & $1/\sqrt{3} \approx 0.5774$\\  
\hline
\end{tabular}
\caption{Values of the constants $\kappa$ and $\delta^*$ in \eqref{delta:kappa} for various recovery algorithms.}
\label{tableRIP}
\end{center}
\end{table}%

A Gaussian random matrix $\mtx{A} \in \R^{\m \times \n}$ is a matrix that has independent, normally distributed entries with mean zero and variance one. It is shown, e.g., in \cite{cata06,mepato09,badadewa08} that the restricted isometry
constants of $\frac{1}{\sqrt{\m}} \mtx{A}$ satisfy $\delta_\s \leq \delta$ with high probability provided that
\[
m ~\geq~ C \delta^{-2} \s \log(\n/\s).
\]
It follows that the number $\m$ of Gaussian measurements required to reconstruct an $\s$-sparse signal of length $\n$ is \emph{linear} in the sparsity and \emph{logarithmic} in the ambient dimension.
See \cite{cata06,mepato09,badadewa08,foraXX,ra10} for precise statements and extensions to Bernoulli and subgaussian matrices.
It follows from lower estimates of Gelfand widths that this bound on the required
samples is optimal \cite{codade09,foparaul10,gagl84}; that is, the $\log$-factor must be present.

For a matrix consisting of $\m$ random rows from an $\n \times \n$ discrete Fourier transform matrix, slightly weaker estimates are available~\cite{cata06,rudelson08sp,ra08,ra10}.
The restricted isometry constants of this matrix satisfy $\delta_\s \leq \delta$ with high probability provided that
\[
\m ~\geq~ C \delta^{-2} \s \log^3(\s) \log(\n).
\]

\subsection{Partial Random Circulant Matrices}

Given a vector $\vct{\phi} = (\phi_0,\hdots,\phi_{\n-1})^T \in \R^\n$, we introduce the circulant matrix
\begin{equation}
	\label{eq:Phidef:zero}
	\Fee^\circ ~= ~
	\begin{bmatrix}
		\phi_0 & \phi_{\n-1} & \cdots & \phi_1 \\
		\phi_1 & \phi_0 & \cdots & \phi_2 \\
		\vdots & \vdots & & \vdots \\
		\phi_{\n-1} & \phi_{\n-2} & \cdots & \phi_0 \\ 
	\end{bmatrix} \; \in \R^{\n \times \n}.
\end{equation}
Square matrices are not very interesting for compressed sensing, so we will restrict our attention to a row submatrix of $\Fee^\circ$.
Consider an \emph{arbitrary} index set $\OB \subset \{0,1,\dots,\n-1\}$ whose cardinality  $|\OB| = \m$.  We define the operator $\ROB\in\R^{\m\times\n}$ that restricts a vector to the entries listed in $\OB$. Then the corresponding partial circulant matrix generated with $\vct{\phi} \in \R^{\n}$ 
is defined as
\begin{equation}
\label{eq:Phidef}
\Fee ~=~ \frac{1}{\sqrt{\m}}\ROB\Fee^\circ.
\end{equation}
The action of $\Fee$ can be interpreted as a circular convolution with the sequence 
$\frac{1}{\sqrt{\m}}\vct{\phi}$ followed by a subsampling at locations indexed by $\Omega$.  

We will demonstrate that a partial circulant matrix with a random generator $\vct{\phi}$ has small restricted isometry constants.  As a result, we can recover a sparse vector $\vct{x}$ robustly from measurements $\vct{y} = \Fee \vct{x}$ using any of the algorithms mentioned above.  Since $\Fee^{\circ}$ can be diagonalized via the Fourier transform, the matrices $\Fee$ and $\Fee^*$ both admit fast matrix--vector multiplication using the FFT algorithm.  This fact allows us to accelerate recovery algorithms substantially.

A Rademacher sequence $\eps = (\varepsilon_1,\hdots,\varepsilon_\n)^T$ is a sequence of 
independent random variables, each taking the values $+1$ and $-1$ with
equal probability. In the sequel, the matrix $\Fee$ in \eqref{eq:Phidef} will always be generated by a Rademacher sequence $\vct{\phi} = \eps$, and we will refer to it as a \emph{partial random circulant matrix}.

The main result of this paper is the following theorem.
\begin{theorem}
	\label{th:mean}
	Let $\OB$ be an \emph{arbitrary} subset of $\{0,1,\ldots,\n-1\}$ with cardinality $|\OB|=\m$.
	Let $\Fee$ be the corresponding partial random circulant matrix \eqref{eq:Phidef} generated by a Rademacher sequence, and let $\delta_\s$
	denote the $\s$th restricted isometry constant. 
	Then 
	\begin{equation}
		\label{eq:Ebound}
		\E[\delta_\s] ~\leq~ C_1\max\left\{\frac{\s^{3/2}}{\m}\log^{3/2}\n,~
		\sqrt{\frac{\s}{\m}}\log \s \log \n\right\} 
	\end{equation}
	where $C_1>0$ is a universal constant.  
\end{theorem}

In particular, \eqref{eq:Ebound} implies that for given $\delta \in (0,1)$, we have $\E[\delta_\s]\leq \delta$ provided
	\begin{equation}
		\label{eq:m-mean}
		\m ~\geq~ C_2\max\left\{\delta^{-1}\s^{3/2}\log^{3/2}\n,~ 
		\delta^{-2}\s\log^2\n \log^2\s \right\},
	\end{equation}
where $C_2 >0$ is another universal constant.

Theorem~\ref{th:mean} tells us that partial random circulant matrices $\Phi$ obey \eqref{eq:RIPs} in expectation.  The following theorem states 
that the random variable $\delta_\s$ does not deviate much from its mean. 
\begin{theorem}
	\label{th:tail}
	Let $\delta_\s$ be as in Theorem~\ref{th:mean}.  Then for $0\leq\lambda\leq 1$
	\[
		\P\left(\delta_\s \geq \E[\delta_\s] + \lambda\right)
		~\leq~
		e^{-\lambda^2/\sigma^2}
		\quad\text{where}\quad\sigma^2 = C_3\frac{\s}{\m}\log^2\s \log^2 \n,
	\]
	for a universal constant $C_3 > 0$.
\end{theorem}

The proof of Theorem \ref{th:mean} is connected with the approach for random partial Fourier matrices \cite{rudelson08sp,ra10}.  
We use a version of the classical Dudley inequality for Rademacher chaos that bounds the expectation of its supremum by the maximum of two entropy integrals that involve covering numbers with respect to two different metrics.  We use elementary ideas from Fourier analysis to provide bounds for these metrics.  This reduction allows us to exploit covering number estimates from the RIP analysis for partial Fourier matrices~\cite{rudelson08sp,ra10} to complete the argument.

\subsection{Discussion}

In essence, the bound \eqref{eq:m-mean} exhibits the scaling behavior $\m\gtrsim\s^{3/2}\log^{3/2}\n$.  This result improves on the best available result for this type of matrix~\cite{haupt10to}, but it falls short of the linear scaling in $\s$ that is typical in the compressive sensing literature.  The bottleneck in our argument appears to be the bound on the ``subexponential integral'' (Section~\ref{sec:d1integral}).  It is not clear how to significantly improve \eqref{eq:d1bound} or the covering numbers from Proposition~\ref{prop:rvcover}, so tightening this bound to $m\gtrsim\s\log^p\n$ for some constant $p$ will probably require a different approach.
Indeed, it is known that the central tool in this paper, the Dudley-type inequality for Rademacher chaos stated in Proposition~\ref{prop:Dudley:chaos}, is not sharp for all examples \cite{leta91,ta05-2}.  It might be that we are facing one of these cases. 

The statement of Theorem~\ref{th:mean} uses a very specific model for the measurement matrix based on a partial random convolution with a generator $\vct{\phi}$ given by a Rademacher sequence.
We have restricted our discussion to this example to simplify the exposition.  Analogous results for other types of random generator sequences can be derived using the same type of analysis.  In particular, one might consider the following variations.
\begin{description}
	\item[Gaussian generating sequence.]  We can take the sequence $\vct{\phi}$ to be iid Gaussian with 
	zero mean and variance one.  In this case, we can establish \eqref{eq:m-mean} by repeating the same steps because the central tool, Proposition~\ref{prop:Dudley:chaos}, holds for Gaussian chaos processes as well as Rademacher chaos processes (perhaps with different constants).  It is possible that, for this case, there is an extra factor of $\log\n$ in the denominator of the variance $\sigma^2$ in the tail bound in Theorem~\ref{th:tail}.
	
	\item[Fourier-domain randomness.]  The generating sequence can also be iid Bernoulli in the Fourier domain.  That is, we can take $\vct{\phi} = \sqrt{\n} \FT^{-1}\eps$, where $\FT$ is the Fourier matrix (see below) and $\eps$ is a Rademacher sequence.
The analysis in this case is almost identical, except that we take the Fourier-domain expression \eqref{eq:Gx} for the random process as our starting point. 

	This type of model was analyzed in \cite{romberg09co} for the case where also $\OB$ is chosen randomly; Theorem~\ref{th:mean} gives us a result when $\OB$ is \emph{arbitrary}.  Our model is also related to the {\em random demodulator} system analyzed in \cite{tropp10be}.  If we switch the roles of time and frequency, we can interpret the measurement system $\Fee$ as taking a signal that is sparse in the Fourier domain, multiplying it pointwise by a Rademacher sequence in the time domain (the random demodulation), and then recording the frequency components indexed by the set $\OB$.  If $\OB$ consists of a sequence of consecutive indices, then this operation is equivalent with random demodulation followed by bandpass filtering and, finally, acquiring $\m$ uniformly spaced samples.  (In our model, we are observing the Fourier transform of the samples rather than the samples themselves.)  This observation broadens the ``randomly demodulate, integrate, then subsample'' architecture of \cite{tropp10be} to ``randomly demodulate, bandpass filter, then subsample''.
	
	\item[Complex generating sequences.]  It is also possible to take the generating sequence to be either a complex Steinhaus or complex Gaussian sequence.  The proofs above remain essentially the same, the main difference would be establishing complex versions of Proposition~\ref{prop:Dudley:chaos} and Theorem~\ref{th:boucheron} (some related results for Steinhaus sequences can be found in \cite[Ch.\ 4]{ra10}).
	
	\item[Toeplitz matrices.]  We can obtain analogous results for sections of a random Toeplitz matrix because a Toeplitz matrix can be embedded in a circulant matrix of twice the dimension.
\end{description}

{\bf Applications.}
From an engineering point of view, Theorems~\ref{th:mean} and \ref{th:tail} tell us that we can identify a system with a sparse impulse response by probing it with a random input sequence and then taking a small number of samples of the output.  This type of system identification, or \emph{deconvolution}, problem is common task in signal processing, and the fact that it can be performed from a small number of samples allows for some interesting new design considerations.  

For example, in radar imaging a transmitter sends out a pulse, which reflects off of a number of targets, and then a receiver observes this superposition of pulses (which can be modeled as the original pulse convolved with an unknown sparse range profile).  The resolution to which we can resolve target locations is determined by the bandwidth of the pulse; to reconstruct the range profile digitally, then, the system requires an analog-to-digital converter (ADC) whose sample rate is on the same order as this bandwidth.  Typical pulse bandwidths are in the gigahertz range, and ADCs that operate at this rate are expensive and low resolution.  Indeed, lack of good high-speed ADCs has ``historically slowed the introduction of digital techniques into radar signal processing'' \cite[Chap. 1.2]{richards05fu}.  Theorem~\ref{th:mean} suggests that sample rate of the ADC depends primarily on the sparsity of the range profile, rather than the bandwidth of the pulse, which would allow us to achieve the same resolution with less expensive and more accurate hardware.  See~\cite{romberg09co,herman09hi} for more discussion of how convolution with a random waveform followed by subsampling can be applied to active imaging problems.

Another application of sparse recovery from a random convolution is increasing the field-of-view of a camera using a coded aperture \cite{marcia08fa,gottesman89ne}.  Here, we can imagine an optical architecture where a large image is convolved with a random code and then a small spatial portion of the result is sampled on a compact pixel array.  If the image is sparse enough, Theorem~\ref{th:mean} suggests that the entire field-of-view can be reconstructed to full resolution from this small set of observations. 

{\bf Dimensionality reduction.}
The Johnson--Lindenstrauss lemma is an important tool for dimensionality reduction. It establishes
that the pairwise distances between points in a high-dimensional space are approximately
preserved after we project the points into a significantly lower-dimensional space using a random linear map.
While Gaussian or Bernoulli matrices were initially used for this task,
more recent analyses show that structured random matrices also work.
In particular, Hinrichs and Vyb{\'i}ral \cite{hivy10} have shown that one can
perform dimension reduction using partial random circulant matrices with
randomized column signs.
These matrices are computationally efficient because they can be applied
using the FFT algorithm. 
Krahmer and Ward subsequently showed that
a matrix satisfying the RIP provides a Johnson--Lindenstrauss embedding
if one randomizes the column signs \cite{krwa10}. 
Together with our result on the RIP of partial random
circulant matrices, the work of Krahmer and Ward improves on a related result
by Vyb{\'i}ral \cite{vy10}.  See \cite{krwa10} for a precise statement.

\subsection{Relationship with previous work}

Numerical results for compressive sampling by random convolution followed by subsampling appear in~\cite{tropp06fi}.  In this paper, the measurement process convolves with a pulse of length $B$ and then extracts $\m$ equally spaced samples from the convolution.  The effectiveness of this strategy is quantified empirically as a function of pulse length and the undersampling ratio: for long enough pulses, the number of samples required to reconstruct a signal is approximately linear in the sparsity.

Later, theoretical results for compressed sensing using random convolution were developed in~\cite{romberg07co,romberg09co}.  In these works, the measurement model is slightly different; the generating sequence $\vct{\phi}$ is the discrete Fourier transform of an iid sequence of random signs.
Convolution with this spectrally random sequence is followed by sampling at \emph{random} locations (as opposed to the arbitrary set $\OB$ we are considering in this work).  This process is universally efficient, in that an $\s$-sparse signal can be reconstructed from $\m\gtrsim\s\log^p\n$ samples independent of the orthobasis in which it is sparse.

The first theoretical results for the measurement model we are using in this paper, convolution with a iid sequence followed by subsampling at fixed locations, can be traced to~\cite{bajwa07to,haupt10to}.  They show that the matrix $\Fee$ define in \eqref{eq:Phidef} has the RIP of order $\s$ with high probability when $m \gtrsim \s^2\log\n$.  These works are couched in the language of channel estimation, and so the results are stated explicitly for the case where $\OB$ contains consecutive indices.  Nevertheless, it appears that the same proof strategy extends to arbitrary $\OB$.  Theorems~\ref{th:mean} and \ref{th:tail} refine the sufficient condition for this model to $m\gtrsim \s^{3/2}\log\n$ using a completely different mathematical analysis.

A nonuniform recovery result for partial random circulant matrices has been established in \cite{rauhut09ci,ra10}.  (See \cite{ra10} for a discussion of the difference between nonuniform and uniform recovery guarantees.) 
Suppose that the number $\m$ of measurements satisfies $\m\gtrsim\s\log^2\n$.  Let $\vct{x}_0$ be an $\s$-sparse vector $\vct{x}_0$ whose nonzero components have random signs. With high probability, we can recover this vector exactly via $\ell_1$-minimization using the measurements $\vct{y} = \Fee \vct{x}_0$, where the partial random circulant matrix $\Fee$ is drawn independently from $\vct{x}_0$.  The proof involves duality for convex optimization, and it does not establish any type of RIP.  As a consequence, this work does not offer any guarantees about stability in the presence of measurement noise or robustness when the signal $\vct{x}_0$ is not exactly $\s$-sparse, in contrast with the RIP recovery bound~\eqref{rec:stability}.

\section{Proof of Theorem~\ref{th:mean} (Expectation)}
\label{sec:MeanProof}

We develop a method for estimating the restricted isometry constant $\delta_\s$ for a fixed sparsity level $\s$.  Let $T$ denote the set of all $\s$-sparse signals in the Euclidean unit ball:
\begin{equation} \label{eqn:T-set}
	T := \{ \vct{x} \in\R^\n: \|\vct{x}\|_0\leq\s,~\|\vct{x}\|^2_2 \leq 1\}.
\end{equation}
Define a function $\spnorm{\cdot}$ on Hermitian $\n\times\n$ matrices via the formula
\begin{equation}
	\nonumber
	\spnorm{\mtx{A}} := \sup_{\vct{x}\in T} \left|\vct{x}^*\mtx{A} \vct{x} \right|.
\end{equation}
This function can be extended to a norm on the set of all square matrices.  We work with the quantity $\spnorm{\Fee^*\Fee-\mathbf{I}}$ because
\begin{equation}
	\label{eq:snormip}
	\spnorm{\Fee^*\Fee-\mathbf{I}}
	~=~
	\sup_{\vct{x}\in T} \left|\<(\Fee^*\Fee - \mathbf{I})\vct{x},\vct{x}\>\right|
	~=~
	\sup_{\vct{x}\in T}\left|~\|\Fee\vct{x}\|^2_2 - \|\vct{x}\|^2_2 ~\right|
	~=~ \delta_\s.
\end{equation}

Let $\Sh$ be the cyclic shift down operator on column vectors in $\R^\n$.  Applying the power $\Sh^k$ to $\vct{x}$ cycles $\vct{x}$ downward by $k$ coordinates: $(\Sh^k \vct{x})_\ell = x_{\ell \ominus k}$, where $\ominus$ is subtraction modulo $\n$.  Note that $(\Sh^k)^* = \Sh^{-k}=\Sh^{\n-k}$.  We can now express $\Fee$ as a random sum of shift operators,
\[
	\Fee = \frac{1}{\sqrt{\m}}\sum_{k=1}^{\n}\ek \ROB \Sh^k.
\]
It follows that
\begin{equation}
	\label{eq:PhitPhi-I}
	\Fee^*\Fee - \mathbf{I} ~=~ 
	\frac{1}{\m}\sum_{k\neq\ell} \ek\el \, \Sh^{-k}\ROB^*\ROB \Sh^{\ell} 
	~=~ \frac{1}{\m}\sum_{k\neq\ell} \ek\el \, \Sh^{-k}\mtx{P}_\OB \Sh^{\ell}, 
\end{equation}
where $\mtx{P}_\OB = \ROB^* \ROB$ is the $\n\times\n$ diagonal projector onto the coordinates in $\OB$.  Applying $\mtx{P}_{\OB}$ to $\vct{x}$ preserves the values of $\vct{x}$ on the set $\OB$ while setting the values outside of $\OB$ to zero.  

Combining \eqref{eq:snormip} and \eqref{eq:PhitPhi-I}, we can view the restricted isometry constant as the supremum of a random process indexed by the set $T$:
\begin{equation} \label{eqn:delta-Gx}
	\delta_\s ~=~ \sup_{\vct{x}\in T} |G_{\vct{x}}| \quad\text{where}\quad
	G_{\vct{x}} ~=~  \frac{1}{\m}\sum_{k\neq\ell} \ek\el \,
	\vct{x}^*\Sh^{-k}\mtx{P}_\OB \Sh^{\ell}\vct{x}
\end{equation}
We must bound the expected supremum of this process.

\subsection{Fourier representation of the random process}

One of the key ideas in this work is to re-express the random process $G_{\vct{x}}$ in the Fourier domain.  Let $\FT$ be the $\n\times\n$ discrete Fourier transform matrix whose entries are given by the expression
\[
	{F}({\omega,\ell}) ~:=~ {\rm e}^{-\iunit 2\pi \, \omega\ell /\n},\quad 0\leq\omega,\ell\leq\n - 1.
\]
Note that we employ the electrical engineering convention that $\FT$ is unnormalized.  The hat symbol indicates the Fourier transform of a vector: $\fx := \FT\vct{x}$.  Recall that a shift in the time domain followed by a Fourier transform can also be written as a Fourier transform followed by a frequency modulation:
\[
	\FT \Sh^k ~=~ \Mod^k \FT,
\]
where $\Mod$ is the diagonal matrix with entries ${M}(\omega,\omega) := {\rm e}^{-\iunit 2\pi \, \omega/\n}$ for $0 \leq \omega \leq \n - 1$.

The random process $G_{\vct{x}}$ has the Fourier-domain representation
\begin{equation}
	\label{eq:Gx}
	G_{\vct{x}} ~=~ \frac{1}{\m}\sum_{k\neq\ell} \ek\el \,
	\fx^*\Mod^{-k}\PO \Mod^{\ell}\fx,
\end{equation}
where $\PO=\n^{-1}\FT \mtx{P}_\OB \FT^{-1}$. 
The matrix $\PO$ has several nice properties that we use in the sequel.

\begin{lemma}
	\label{lem:projector}
	The $\n\times\n$ matrix $\PO = n^{-1} \FT \mtx{P}_{\OB} \FT^{-1}$ has the following properties: 
	\begin{enumerate}
		
		\item $\PO$ is circulant and conjugate symmetric.
		
		\item Along the diagonal $\widehat{P}_{\OB}(\omega,\omega)=\m/\n^2$, and off the diagonal $|\POab|\leq\m/\n^2$.

		\item Since the rows and columns of $\PO$ are circular shifts of one another,
		\[
			\sum\nolimits_{\oa}\left|\POab\right|^2
			~=~ \sum\nolimits_{\ob}\left|\POab\right|^2 
			~=~ \|\PO\|_F^2/\n ~=~ \m/\n^3.
		\]
		
		\item $\PO$ has exactly $\m$ nonzero eigenvalues, each of which is equal to $1/\n$.  As such, $\PO$ has spectral norm $\|\PO\|=1/\n$ and Frobenius norm $\|\PO\|_F^2 = \m/\n^2$.

	\end{enumerate}
\end{lemma}

\begin{proof}
These properties follow almost immediately from the fact that $\mtx{P}_{\OB} = \ROB^* \ROB$ is a diagonal matrix with 0--1 entries.  The matrix $\PO$ inherits conjugate symmetry from $\mtx{P}_{\OB}$.  The matrix $\PO$ is circulant because it is diagonalized by the Fourier transform.  Since we form $\PO$ by applying a similarity transform to $\mtx{P}_{\OB}$, they have the same eigenvalues modulo the scale factor $n^{-1}$.  The remaining points follow from the simple calculations described in the statement of the lemma.
\end{proof}

\subsection{Integrability of chaos processes}

For the next step in the argument, we must rewrite the random process~\eqref{eq:Gx} again.  Let $\vct{\eps} = [\varepsilon_0, \dots, \varepsilon_{\n - 1}]^*$.  The process can now be expressed as a quadratic form:
\begin{equation} \label{eqn:2-chaos}
G_{\vct{x}} ~=~ \< \eps, \mtx{Z}_{\vct{x}} \,\eps \>
\quad\text{where}\quad \vct{x} \in T.
\end{equation}
The matrix $\mtx{Z}_{\vct{x}}$ has entries
\[
	Z_{\vct{x}}(k,\ell) ~=~ \begin{cases} 
	\m^{-1}\fx^* \Mod^{-k}\PO \Mod^{\ell}\fx, & k\neq\ell \\
	0, & k=\ell
	\end{cases}.
\]
A short calculation verifies that this matrix can be written compactly.
\begin{equation}
	\label{eq:Zx}
	\mtx{Z}_{\vct{x}} ~=~ \frac{1}{\m}\left(\FT^*\fX^*\PO\fX \FT - \operatorname{diag}(\FT^*\fX^*\PO\fX \FT) \right),
\end{equation}
where $\fX := \operatorname{diag}(\fx)$ is the diagonal matrix constructed from the vector $\fx$.  The term \emph{homogeneous second-order chaos} is used to refer to a random process $G_{\vct{x}}$ of the form~\eqref{eqn:2-chaos} where each matrix $\mtx{Z}_{\vct{x}}$ is conjugate symmetric and hollow, i.e., has zeros on the diagonal.

To bound the expected supremum of the random process $G_{\vct{x}}$ over the set $T$, we apply a version of Dudley's inequality that is specialized to this setting.  Define two pseudo-metrics on the index set $T$:
$$
d_1(\vct{x}, \vct{y}) ~:=~ \| \mtx{Z}_{\vct{x}} - \mtx{Z}_{\vct{y}} \|
\quad\text{and}\quad
d_2(\vct{x}, \vct{y}) ~:=~ \| \mtx{Z}_{\vct{x}} - \mtx{Z}_{\vct{y}} \|_{\rm F}.
$$
Let $N( T, d_i, u )$ denote the minimum number of balls of radius $u$ in the metric $d_i$ that we need to cover the set $T$.

\begin{proposition}[Dudley's inequality for chaos] \label{prop:Dudley:chaos}
Suppose that $G_{\vct{x}}$ is a homogeneous second-order chaos process indexed by a set $T$.  Fix a point $\vct{x}_0 \in T$.  There exists a universal constant ${\rm K}$ such that
	\begin{equation}
		\label{eq:dudley-chaos}
		\E\sup_{\vct{x} \in T} |G_{\vct{x}}-G_{\vct{x}_0}|
			~\leq~ {\rm K} \max\left\{ 
		\int_0^\infty \log N(T,d_1,u)\,{\rm d}u,
		\int_0^\infty \sqrt{\log N(T,d_2,u)}\,{\rm d}u
	     \right\}.
	\end{equation}
\end{proposition}

Proposition~\ref{prop:Dudley:chaos} is based on the idea that the random process has a subexponential part, whose variation is controlled by the integral with respect to $d_1$, and a subgaussian part, whose variation is controlled by the integral with respect to $d_2$.  This result appears in \cite[Thm. 11.22]{leta91} and \cite[Thm. 2.5.2]{ta05-2}.  Our statement of the proposition looks different from the versions presented in the literature, so we sketch the derivation in Appendix~\ref{apx:DudleyChaos}.

\subsection{The subgaussian integral}
\label{sec:d2integral}

In this section, we develop an estimate for the second integral in~\eqref{eq:dudley-chaos}.  To do so, we need a simpler bound for the metric $d_2$.  First, note that
\begin{align*}
d_2(\vct{x}, \vct{y})
	&~\leq~ \frac{1}{\m} \|\FT^*(\fX^*\PO\fX - \fY^*\PO\fY)\FT\|_{\rm F}
	~=~ \frac{\n}{\m} \|\fX^*\PO\fX - \fY^*\PO\fY\|_{\rm F} \\
	&~=~ \frac{\n}{2\m} \|(\fX - \fY)^* \PO (\fX + \fY) +
		(\fX + \fY)^* \PO (\fX - \fY) \|_{\rm F} \\
	&~=~ \frac{\n}{2\m} \left[\sum\nolimits_{\oa,\ob}
	|\POab|^2\cdot|\widehat{Q}_{\vct{x},\vct{y}}(\oa,\ob)|^2\right]^{1/2}
\end{align*}
where we have written
$$
	\widehat{Q}_{\vct{x},\vct{y}}(\oa,\ob) ~:=~ (\fxa+\fya)(\fxb-\fyb)^* + (\fxa-\fya)(\fxb+\fyb)^*.
$$
The first inequality arises when we re-introduce the diagonal entries of the hollow matrices.  The next identity follows from the unitary invariance of the Frobenius norm.  The second line is the polarization identity, and we obtain the last line by expressing the Frobenius norm in terms of coordinates.

Define the $\|\cdot\|_\hi$ norm to be the $\ell_\infty$ norm in the discrete Fourier domain
\[
	\|\vct{x}\|_\hi ~:=~ \|\fx\|_\infty.
\]
We can bound the entries of $\widehat{\mtx{Q}}_{\vct{x},\vct{y}}$ in terms of this norm.  If we abbreviate $\vct{v} ~=~ \vct{x} + \vct{y}$,
$$
	|\widehat{Q}_{\vct{x},\vct{y}}(\oa,\ob)|
 	~\leq~ \|\vct{x}-\vct{y}\|_\hi\cdot(|\widehat{v}(\oa)| + |\widehat{v}(\ob)|).
$$
Introduce the latter bound into our estimate for the metric and apply the triangle inequality to reach
$$ 	\label{eq:d2a}
d_2(\vct{x},\vct{y})
	~\leq~ \frac{\n\|\vct{x}-\vct{y}\|_\hi}{2\m} \left[
	\left( \sum\nolimits_{\oa,\ob}|\POab|^2 \, |\widehat{v}(\oa)|^2 \right)^{1/2} + 
	\left( \sum\nolimits_{\oa,\ob}|\POab|^2 \, |\widehat{v}(\ob)|^2 \right)^{1/2} \right].
$$
Let us examine the first sum more closely.
$$
\left( \sum\nolimits_{\oa,\ob} |\POab|^2 \, |\widehat{v}(\oa)|^2 \right)^{1/2}
	~=~ \left( \frac{\m}{\n^3} \| \widehat{\vct{v}} \|_2^2 \right)^{1/2}
	~\leq~ \sqrt{\frac{\m}{\n^3}} ( \| \widehat{\vct{x}} \|_2 + \| \widehat{\vct{y}} \|_2 )
	~=~ \frac{2\sqrt{\m}}{\n}.
$$
The first identity follows from Point 3 of Lemma~\ref{lem:projector}.  The second relation is the the triangle inequality.  The last identity is a consequence of the fact that $\vct{x}$ and $\vct{y}$ have unit energy together with Parseval's identity.  An analogous argument applies to the second sum.  We conclude that
$$
d_2(\vct{x},\vct{y}) \leq \frac{2 \| \vct{x} - \vct{y} \|_{\hi}}{\sqrt{m}}
	= 2 \sqrt{\frac{s}{m}} \cdot \frac{1}{\sqrt{s}} \| \vct{x} - \vct{y} \|_{\hi}.
$$

This bound on $d_2$ allows us to estimate the subgaussian integral in terms of the covering numbers of $T$ with respect to the norm $s^{-1/2} \| \cdot \|_{\hi}$.  Abbreviating $\alpha = 2\sqrt{s/m}$, we compute that
\begin{align*}
I_2 ~:=&~ \int_0^\infty \sqrt{\log N(T, d_2, u)} \, {\rm d}u
	~\leq~ \int_0^\infty \sqrt{\log N(T, \alpha s^{-1/2} \| \cdot \|_{\hi}, u)}
		\, {\rm d}u \\
	=&~ \int_0^\infty \sqrt{\log N(T, s^{-1/2} \| \cdot \|_{\hi}, \alpha^{-1} u)}
		\, {\rm d}u
	~=~ \alpha \int_0^\infty \sqrt{\log N(T, s^{-1/2} \| \cdot \|_{\hi}, u)}
		\, {\rm d}u \\ 
	=&~ 2\sqrt{\frac{s}{m}} \int_0^1 \sqrt{ \log N(T, s^{-1/2} \| \cdot \|_{\hi}, u ) } \, {\rm d}u.
\end{align*}
The first inequality uses the fact that the metric balls in $d_2$ are \emph{larger} than the balls in the norm $\alpha s^{-1/2} \| \cdot \|_{\hi}$ because the metric is \emph{smaller} than the norm.  The second line follows from an elementary scaling property of covering numbers along with a change of variables in the integral.  Finally, we apply the fact that $T$ is contained in the unit ball of $s^{-1/2} \| \cdot \|_{\hi}$ to see that the integrand vanishes for $u \geq 1$.
We can now exploit some covering number estimates that appear in the literature~\cite{ra10,rudelson08sp}.  The first bound follows from a volume comparison argument; the second uses the empirical method invented by Maurey \cite{pi81} and refined by Carl \cite{ca85}.

\begin{proposition}[Covering Numbers]
	\label{prop:rvcover}
	For $u \in (0,1]$, we have the following bound.
	\begin{equation}
		\label{eq:rvcover}
		N(T, s^{-1/2} \|\cdot\|_\hi,u) ~\leq~
		\min\left\{ \left(\frac{{\rm C}_1\n}{\s}\right)^\s(1+2/u)^\s,~
		\n^{{\rm C}_2(\log\n)/u^2}\right\}.
	\end{equation}
	The ${\rm C}_i$ are positive universal constants.
\end{proposition}	
Explicit values of the constants $C_1,C_2$ can be found in \cite[Lem.\ 8.3, eq.\ (8.14)]{ra10}.
We finish off the estimate for the first integral using Proposition~\ref{prop:rvcover}.  Splitting the integral at $\lambda$, we have
\begin{align} \label{eqn:subgauss-int}
I_2	~&\leq~ {\rm C} \sqrt{\frac{\s}{\m}} \left[
	\int_0^\lambda \sqrt{\s(\log(\n/\s) + \log(1+2/u))} \, {\rm d}u ~+~
	\log(\n) \int_\lambda^1 u^{-1}~{\rm d}u \right] \notag \\
	&\leq~ {\rm C} \sqrt{\frac{\s}{\m}} \left[
		\lambda \sqrt{\s \log(\n/\s)} + \lambda \sqrt{\s} \log(\lambda^{-1}) +
		\log(\n) \log(\lambda^{-1}) \right] \notag \\
	&\leq~ {\rm C} \sqrt{\frac{\s \log^2(\s)\log^2(\n)}{ \m } },
\end{align}
where we have chosen $\lambda = s^{-1/2}$ in the last step.

\subsection{The subexponential integral}
\label{sec:d1integral}

We can also bound the $d_1$ metric in terms of the norm $\|\cdot\|_\hi$, which allows us to re-use the estimates for the covering numbers given in Proposition~\ref{prop:rvcover} to control the first integral in~\eqref{eq:dudley-chaos}.  To begin,
\[
	d_1(\vct{x},\vct{y}) = \|\mtx{A}_{\vct{x},\vct{y}}-\mtx{D}_{\vct{x},\vct{y}}\|
		\leq \|\mtx{A}_{\vct{x},\vct{y}}\| + \|\mtx{D}_{\vct{x},\vct{y}}\|
\]  
where the matrix $\mtx{A}_{\vct{x},\vct{y}}$ is given by the expression 
\begin{align} \label{eqn:Axy}
	\mtx{A}_{\vct{x},\vct{y}} :=&~
	\frac{1}{\m}\FT^*(\fX^*\PO\fX - \fY^*\PO\fY)\FT \notag \\
	=&~ \frac{1}{2\m} \FT^*((\fX+\fY)^*\PO(\fX-\fY) +  (\fX-\fY)^*\PO(\fX+\fY))\FT,
\end{align}
and $\mtx{D}_{\vct{x},\vct{y}}$ denotes the diagonal of the matrix $\mtx{A}_{\vct{x},\vct{y}}$.

We bound the diagonal term first.  Let $\vct{f}_k$ be the $k$th column of $\FT$, and note that $\| \vct{f}_k \|_2^2 = \n$.  Owing to Lemma~\ref{lem:projector}, we have
\begin{align*}
	\|\mtx{D}_{\vct{x},\vct{y}}\| &= 
	\frac{1}{\m}\max_k\left|\vct{f}_k^*(\fX+\fY)^*\PO(\fX-\fY)\vct{f}_k\right| \\
	~&\leq~ \frac{1}{\m}\max_k\|\PO(\fX+\fY)\vct{f}_k\|_2\cdot \|(\fX-\fY)\vct{f}_k\|_2 \\
	&\leq \frac{1}{\m}\max_k \|\PO\|\cdot\|\fX+\fY\|\cdot\|\fX-\fY\|\cdot\|\vct{f}_k\|^2_2 
	~=~ \frac{1}{\m}\|\vct{x}+\vct{y}\|_\hi\cdot\|\vct{x}-\vct{y}\|_{\hi} \\
	&\leq \frac{2\s}{\m}\cdot \frac{1}{\sqrt{s}} \|\vct{x}-\vct{y}\|_\hi.
\end{align*}
In the last inequality, we have used the fact that $\|\vct{x}+\vct{y}\|_\hi\leq\|\vct{x}+\vct{y}\|_1\leq 2\sqrt{\s}$ for $\vct{x},\vct{y}\in T$.
For the off-diagonal term, we use Lemma~\ref{lem:projector} to compute 
\begin{align*}
	\|\mtx{A}_{\vct{x},\vct{y}}\| &= \frac{1}{\m}\|\FT^*(\fX+\fY)^*\PO(\fX-\fY)\FT\| 
	~\leq~ \frac{\n}{\m}\|(\fX+\fY)^*\PO(\fX-\fY)\| \\
	&\leq \frac{\n}{\m}\|(\fX+\fY)\|\cdot\|\PO\|\cdot \|(\fX-\fY)\|
	~=~  \frac{1}{\m}\|\vct{x}+\vct{y}\|_\hi\cdot\|\vct{x}-\vct{y}\|_\hi \\
	&\leq \frac{2\s}{\m}\cdot \frac{1}{\sqrt{s}} \|\vct{x}-\vct{y}\|_\hi.
\end{align*}
In summary,
\begin{equation}
	\label{eq:d1bound}
	d_1(\vct{x},\vct{y}) \leq \frac{4\s}{\m} \cdot \frac{1}{\sqrt{s}} \|\vct{x}-\vct{y}\|_\hi.
\end{equation}

The covering number estimates of Proposition~\ref{prop:rvcover} allow us to bound the subexponential integral.
\begin{align} \label{eqn:subexp-int}
	\nonumber
	I_1 ~:=&~
	\int_0^\infty \log N(T,d_1,u)~{\rm d}u
	~\leq~ \frac{4\s}{\m} \int_0^1 N(T, s^{-1/2} \|\cdot\|_\hi, u)~{\rm d}u \\
	\nonumber
	\leq&~ \frac{{\rm C} \s}{\m}\left(\int_0^\lambda (\s \log(\n/\s) + \s \log(1+2/u)) \, {\rm d}u ~+~
	\log^2(\n) \int_\lambda^1 u^{-2}\,{\rm d}u\right) \notag\\
	\leq&~ \frac{{\rm C} \s}{\m}\left(\lambda\s \log(\n/\s)+ \lambda \s \log(1+2/\lambda) + \lambda^{-1}\log^2(\n) \right) \notag \\
	\leq&~\frac{{\rm C} \s^{3/2}\log^{3/2}(\n)}{\m}.
\end{align}
We have taken $\lambda = \s^{-1/2}\log^{1/2}(\n)$ in the last step.

\subsection{Denouement}

As we noted in~\eqref{eqn:delta-Gx}, the restricted isometry constant $\delta_s$ is given by the supremum of the random process $G_{\vct{x}}$.  To compute the expectation of this supremum, we simply apply Proposition~\ref{prop:Dudley:chaos}.  Select $\vct{x}_0 = \vct{0}$ so that $G_{\vct{x}_0} = 0$.  Introduce the estimate~\eqref{eqn:subexp-int} for the subexponential integral and~\eqref{eqn:subgauss-int} for the subgaussian integral into Dudley's inequality~\eqref{eq:dudley-chaos}.
$$
\E \delta_s ~=~ \E \sup_{\vct{x} \in T} | G_{\vct{x}} |
	~\leq~ {\rm K} \left[ \frac{\s^{3/2} \log^{3/2}(\n)}{\m} + \sqrt{\frac{\s \log^2(\s)\log^2(\n)}{\m}} \right].
$$

This point completes the proof.

\section{Proof of Theorem~\ref{th:tail} (Tail Bound)}
\label{sec:TailProof}

In this section, we develop a tail bound on the supremum of the process $G_x$.  We require the following result, which is Theorem 17 from \cite{boucheron03co}.
Let $\mathcal{F}$ be a collection of $\n\times\n$ symmetric real matrices, and assume that $Z(k,k) = 0$ for each $\mtx{Z} \in\mathcal{F}$.
We are concerned with the tail behavior of the real-valued random variable 
 	\[
 		Y~:=~\sup_{\mtx{Z}\in\mathcal{F}}\sum_{k,\ell=1}^\n
		\varepsilon_k\varepsilon_\ell \, Z(k,\ell).
 	\]
 Define two variance parameters:
	\[
		U ~:=~ \sup_{\mtx{Z} \in\mathcal{F}} \|\mtx{Z}\|
 	\]
 and
 	\[
 		V^2 ~:=~ \E\sup_{\mtx{Z}\in\mathcal{F}}\sum_{k=1}^\n\left|\sum_{\ell=1}^\n 
 		\varepsilon_\ell Z(k,\ell)\right|^2.
 	\]
The parameter $V^2$ describes the variance of $X$ near its mean, while the second parameter $U$ is the scale on which large deviations occur.

\begin{proposition}[Tail Bound for Chaos]  
	\label{th:boucheron}
	Under the preceding assumptions,
 	\begin{equation}
 		\label{eq:boucheron}
 		\P\left\{Y \geq \E[Y] + \lambda\right\}
 		~\leq~
 		\exp\left(-\frac{\lambda^2}{32V^2 + 65U\lambda/3} \right)
 	\end{equation}
 	for all $\lambda\geq 0$.
\end{proposition}

Recall from~\eqref{eqn:delta-Gx},~\eqref{eqn:2-chaos},~and~\eqref{eq:Zx} that the restricted isometry constant can be written as
\begin{equation*} \label{eqn:Gx-tail}
\delta_s ~=~ \sup_{\vct{x} \in T} | G_{\vct{x}} | ~=~ \sup_{\vct{x} \in T}
	\sum\nolimits_{k,\ell} \varepsilon_k \varepsilon_\ell \, Z_{\vct{x}}(k, \ell)
\end{equation*}
where the matrix $\mtx{Z}_{\vct{x}}$ has the expression
$$
\mtx{Z}_{\vct{x}} ~=~ \mtx{A}_{\vct{x}} - {\rm diag}(\mtx{A}_{\vct{x}})
\quad\text{for}\quad
\mtx{A}_{\vct{x}} ~=~ \frac{1}{\m} \FT^*\fX^*\PO\fX \FT.
$$
As a consequence, Theorem~\ref{th:boucheron} applies to the random variable $\delta_s$.

To bound the first parameter $U$, we first apply the triangle inequality to obtain
\[
	\|\mtx{Z}_{\vct{x}}\| ~\leq~
	\|\mtx{A}_{\vct{x}}\| ~+~ \|\operatorname{diag}(\mtx{A}_{\vct{x}})\|.
\]
Emulating the arguments in Section~\ref{sec:d1integral}, we can bound each of the two terms.
\[
	\|\mtx{A}_{\vct{x}}\| ~\leq~ \frac{\n}{\m}\|\PO\|\cdot\|\vct{x}\|_\hi^2
	~\leq~ \frac{\s}{\m},
\]
Similarly,
\begin{align*}
	\|\operatorname{diag}(\mtx{A}_{\vct{x}})\|
	~=~ \frac{1}{\m}\max_k \left|\vct{f}_k^*\fX^*\PO\fX \vct{f}_k \right| 
	~\leq~ \frac{\s}{\m}.
\end{align*}
In total, $U \leq 2\s/\m$.

To bound the other parameter $V^2$, we use the following ``vector version'' of the Dudley inequality, which we prove in the Appendix.

\begin{proposition}\label{prop:vectorDudley}
Consider the vector-valued random process
	\[
		\vct{h}_{\vct{x}} ~=~ \mtx{Z}_{\vct{x}} \, \eps
		\quad\text{for $\vct{x} \in T$.}
	\]
Recall the definition of the pseudo-metric
$$
d_2( \vct{x}, \vct{y} ) ~:=~ \| \mtx{Z}_{\vct{x}} - \mtx{Z}_{\vct{y}} \|_{\rm F}
$$
Fix a point $\vct{x}_0\in T$.  There exists a universal constant ${\rm K} > 0$ such that
\begin{equation}
		\label{eq:dudleyvec}
		\left(\E\sup_{\vct{x}\in T}
		\| \vct{h}_{\vct{x}} - \vct{h}_{\vct{x}_0}\|^2_2\right)^{1/2}
		~\leq~
		{\rm K} \int_{0}^{\infty}\sqrt{N(T,d_2,u)}~{\rm d}u.
	\end{equation}
\end{proposition}

With $\vct{x}_0~=~\vct{0}$, the left-hand side of \eqref{eq:dudleyvec} is precisely $V$.  We have already studied the integral on the right-hand side of \eqref{eq:dudleyvec} in Section~\ref{sec:d2integral}.  We import \eqref{eqn:subgauss-int} to reach
\[
	V^2 ~\leq~ \frac{{\rm C}\s}{\m}\log^2(\s) \log^2\n.
\]

We are prepared to complete the tail bound for $\delta_s$.  For $\lambda\leq 1$,
\begin{align*}
		\frac{\lambda^2}{32C(\s/\m)\log^2(\s)\log^2(\n) ~+~ (130/3)(\s/\m)\lambda} 
		&~\leq~
		\frac{1}{{\rm C}'}\min\left(\frac{\lambda^2}{(\s/\m)\log^2(\s)\log^2(\n)},\frac{\lambda}{\s/\m}\right) \\
		&~\leq~ \frac{\lambda^2}{{\rm C}' (\s/\m)\log^2(\s)\log^2(\n)}.
\end{align*}
Applying \eqref{eq:boucheron}, we reach
\[
	\P\left\{\delta_\s > \E[\delta_\s] + \lambda\right\} ~\leq~ {\rm e}^{-\lambda^2/{\rm C}'\sigma^2},
\]
with $\sigma^2 = (\s/\m)\log^2(\s)\log^2(\n)$.

\begin{appendix}

\section{A Dudley-type inequality for chaos processes}
\label{apx:DudleyChaos}

We provide a proof sketch for Proposition \ref{prop:Dudley:chaos}. Let $\{\varepsilon'_k \}$ be a Rademacher sequence independent of $\{\varepsilon_k\}$. The decoupling method (see, for example, \cite[Th.\ 3.1.2]{pena99de}) yields
\begin{align}
	\E \sup_{\vct{x} \in T} |G_{\vct{x}_0} - G_{\vct{x}}| & ~=~ 
	\E \sup_{\vct{x} \in T} |\sum_{k,\ell} \varepsilon_k \varepsilon_\ell (Z_{\vct{x}_0}(k,\ell) - Z_{\vct{x}}(k,\ell))|\notag\\
	&~\leq~ 8 \E \sup_{\vct{x} \in T} |\sum_{k,\ell} \varepsilon_k \varepsilon_\ell' 
	(Z_{\vct{x}_0}(k,\ell) - Z_{\vct{x}}(k,\ell))|.\notag
\end{align}
Now we introduce two independent standard Gaussian sequences $\{ g_k \}$ and $\{g_k'\}$.  Applying the contraction principle~\cite[eq.\ (4.8)]{leta91} twice, first conditioned on $\varepsilon_\ell$ and then on $g_\ell'$, leads to
\begin{align}
	\E \sup_{\vct{x} \in T} |G_{\vct{x}_0} - G_{\vct{x}}| & ~\leq~ 8 \sqrt{\frac{\pi}{2}} 
	\E  \sup_{\vct{x} \in T} |\sum_{k,\ell} \varepsilon_k g_\ell' (Z_{\vct{x}_0}(k,\ell) - Z_{\vct{x}}(k,\ell))|\notag\\
	& ~\leq~ 4\pi \E  \sup_{\vct{x} \in T} |\sum_{k,\ell} g_k g_\ell' (Z_{\vct{x}_0}(k,\ell) - Z_{\vct{x}}(k,\ell))|.
\end{align}
Thus our task is to bound the expected supremum of a decoupled Gaussian chaos process.  Using \cite[Th. 1.2.7]{ta05-2}, we see that
\[
	\E \sup_{\vct{x}\in T} |G_{\vct{x}} - G_{\vct{x}_0}| ~\leq~ {\rm C} (\gamma_1(T,d_2) + \gamma_2(T,d_1)),
\]
where $\gamma_\alpha(T,d)$ is the $\gamma_\alpha$-functional of the metric space $(T,d)$; see~\cite[Def.~1.2.5]{ta05-2}.  It is known that
\[
	\gamma_{\alpha}(T,d) ~\leq~ {\rm C} \int_0^\infty \log^{1/\alpha}(N(T,d,u))~{\rm d}u.
\]
This is established carefully in \cite[p.\ 13]{ta05-2} for the case $\alpha = 2$ and the general case is analogous. The statement of the theorem follows.

\section{A Dudley type inequality for vector-valued Rademacher processes}
\label{apx:vectorDudley}

In this section we give a sketch of the proof of Proposition \ref{prop:vectorDudley}, the vector-valued version of Dudley's inequality.  It is a consequence of the following proposition.

\begin{proposition}\label{prop:Hoeffding:vec} 
Let $\mtx{A}$ be an $m \times n$ matrix with columns $\vct{a}_1,\dots, \vct{a}_n$. For each $u\geq 0$,
\[
	\P\left(\|\sum_{j=1}^\n \varepsilon_j a_j \|_2 \geq \|\mtx{A}\|_{\rm F} \cdot u \right) \leq 
	2 {\rm e}^{-{\rm c} u^2},
\]
where ${\rm c}$ is a universal constant.
\end{proposition}
\begin{proof} 
	It is easily seen that $\E \|\sum_{j=1}^\n \varepsilon_j \vct{a}_j\|_2^2 = \sum_j^\n \|\vct{a}_j\|_2^2 = \|\mtx{A}\|_{\rm F}^2$.  The vector version of Khintchine's inequality, given in \cite[Th. 1.3.1]{pena99de} implies that, for $p\geq 2$,
\[
	\left(\E \|\sum_{j=1}^n \varepsilon_j \vct{a}_j \|_2^p\right)^{1/p} \leq 
	\sqrt{p} \left(\E \|\sum_{j=1}^\n \varepsilon_j \vct{a}_j\|_2^2\right)^{1/2} = \sqrt{p} \|\mtx{A}\|_F.
\] 
This moment growth implies the tail estimate, see e.g.\ \cite[Proposition 6.5]{ra10}
\[
\P(\|\sum_{j=1} \varepsilon_j \vct{a}_j\|_2 \geq {\rm e}^{1/2} \|\mtx{A}\|_{\rm F} \cdot u) \leq {\rm e}^{-u^2},\quad u \geq \sqrt{2},
\]
which yields the conclusion.
\end{proof}

An explicit value of ${\rm c}=1/2$ for the constant above can be found using non-commutative Khintchine inequalities \cite{bu01,ra10}.
 
With this proposition in place, we can prove Proposition~\ref{prop:vectorDudley} as follows.  From 
Proposition \eqref{prop:Hoeffding:vec},
\[
	\P(\|\vct{h}_{\vct{x}} - \vct{h}_{\vct{y}}\|_2 \geq \|\mtx{Z}_{\vct{x}} - \mtx{Z}_{\vct{y}}\|_{\rm F} \cdot u) \leq 2 {\rm e}^{-{\rm c} u^2}\quad\mbox{ for all } \vct{x},\vct{y} \in T.
\] 
This sets us into the position to follow
the standard proof of Dudley's inequality for
scalar-valued subgaussian processes; see
\cite[Theorem 6.23]{ra10} or \cite{azws09,leta91,ta05-2}. One only has to replace
the triangle inequality for the absolute value by the one for $\|\cdot\|_2$ in $\C^\m$. This finally yields the stated conclusion.

\end{appendix}


\end{document}